\documentclass[12pt]{article}    

\usepackage[margin=0.75in]{geometry} 
\usepackage{amsmath,amssymb,amsthm}
\usepackage{color}
\usepackage{relsize}

\newtheorem{theorem}{Theorem}[section]
\newtheorem{proposition}[theorem]{Proposition}
\newtheorem{lemma}[theorem]{Lemma}

\newtheorem{remark}{Remark}

\newcommand{\al}{\alpha}

\newcommand{\ga}{\gamma}
\newcommand{\s}{\sigma}

\newcommand{\be}{\begin{equation}}
\newcommand{\ee}{\end{equation}}
\newcommand{\bea}{\begin{eqnarray}}
\newcommand{\eea}{\end{eqnarray}}
\usepackage[numbers,sort&compress]{natbib}

\numberwithin{equation}{section}
\begin{document}

\title{Asymptotics of the Largest Eigenvalue Distribution of the Laguerre Unitary Ensemble}
\author{Shulin Lyu\thanks{School of Mathematics (Zhuhai), Sun Yat-sen University, Zhuhai 519082, China; e-mail: lvshulin1989@163.com}, Chao Min\thanks{Corresponding author, Chao Min, School of Mathematical Sciences, Huaqiao University, Quanzhou 362021, China; e-mail: chaomin@hqu.edu.cn}\; and Yang Chen\thanks{Department of Mathematics, Faculty of Science and Technology, University of Macau, Macau, China; e-mail: yangbrookchen@yahoo.co.uk}}


\date{\today}
\maketitle
\begin{abstract}
We study the probability that all the eigenvalues of $n\times n$ Hermitian matrices, from the Laguerre unitary ensemble with the weight $x^{\gamma}\mathrm{e}^{-4nx},\;x\in[0,\infty),\;\gamma>-1$, lie in the interval $[0,\alpha]$.
By using previous results for finite $n$ obtained by the ladder operator approach of orthogonal polynomials, we derive the large $n$ asymptotics of the largest eigenvalue distribution function with $\al$ ranging from 0 to the soft edge. In addition, at the soft edge, we compute the constant conjectured by Tracy and Widom [Commun. Math. Phys. {\bf 159} (1994), 151--174], later proved by Deift, Its and Krasovsky [Commun. Math. Phys. {\bf 278} (2008), 643--678]. Our results are reduced to those of Deift et al. when $\gamma=0$.
\end{abstract}

$\mathbf{Keywords}$: Laguerre unitary ensemble; Largest eigenvalue distribution; Asymptotic behavior;

Ladder operators; Fredholm determinant.

$\mathbf{Mathematics\:\: Subject\:\: Classification\:\: 2010}$: 15B52, 41A60, 42C05.

\section{Introduction}
We consider the Laguerre unitary ensemble (LUE for short) of $n\times n$ Hermitian matrices whose eigenvalues have the following joint probability density function \cite{Mehta}
$$
p(x_1,x_2,\ldots,x_n)=\frac{1}{Z_n}\prod_{1\leq i<j\leq n}(x_i-x_j)^2\prod_{k=1}^{n}w(x_k;\gamma,n),
$$
where $w(x;\gamma,n)$ is the scaled Laguerre weight
\[w(x;\gamma,n)=x^{\gamma}\mathrm{e}^{-4nx},\qquad\;x\in[0,\infty),\quad\gamma>-1,\]
and $Z_n$ is the partition function which reads
\[Z_n:=\int_{[0,\infty)^{n}}\prod_{1\leq i<j\leq n}(x_i-x_j)^2\prod_{k=1}^{n}w(x_k;\gamma,n)dx_{k}.\]

The probability that all the eigenvalues in this LUE lie in the interval $[0,\alpha]$, or the largest eigenvalue is not greater than $\alpha$, is given by
\begin{align}\label{probDn}
\mathbb{P}(n,\gamma,\alpha)=\frac{D_{n}(\alpha)}{D_{n}(\infty)},
\end{align}
where $D_n(\alpha)$ is defined by
\begin{align*}
D_{n}(\alpha):=&\frac{1}{n!}\int_{[0,\alpha]^{n}}\prod_{1\leq i<j\leq n}(x_i-x_j)^2\prod_{k=1}^{n}w(x_k;\gamma,n)dx_{k}.
\end{align*}
It is apparent that $D_n(\infty)=Z_n/n!$.

In this paper, we are interested in the asymptotic behavior of $\mathbb{P}(n,\gamma,\alpha)$ at the soft edge. Deift, Its and Krasovsky \cite{Deift} studied the special case $\mathbb{P}(n,0,\alpha)$, namely the largest eigenvalue distribution on $[0,\alpha]$ of LUE with the weight $\mathrm{e}^{-4nx}$.  By using the Riemann-Hilbert approach, they obtained the constant conjectured by Tracy and Widom \cite{TW1}, which appears in the asymptotic formula for $\mathbb{P}(n,0,\alpha)$ at the soft edge. We would like to generalize their results to general $\gamma$.

By changing variables $4nx_{\ell}=y_{\ell}, \ell=1,2,\ldots,n$, in $D_n(\alpha)$, we get
\bea
D_{n}(\alpha)&=&(4n)^{-n(n+\ga)}\cdot\frac{1}{n!}\int_{[0,4n\al]^{n}}\prod_{1\leq i<j\leq n}(y_i-y_j)^2\prod_{k=1}^{n}y_{k}^{\ga}\mathrm{e}^{-y_{k}}dy_{k}\nonumber\\
&=:&(4n)^{-n(n+\ga)}\widehat{D}_n(4n\al),\nonumber
\eea
where $\widehat{D}_n(\cdot)$ is defined by
\be\label{dnt}
\widehat{D}_{n}(t):=\frac{1}{n!}\int_{[0,t]^{n}}\prod_{1\leq i<j\leq n}(x_i-x_j)^2\prod_{k=1}^{n}x_k^{\ga}\mathrm{e}^{-x_k}dx_{k}.
\ee
It follows from \eqref{probDn} that
\begin{align}\label{probDn1}
\mathbb{P}(n,\gamma,\alpha)=\frac{\widehat{D}_{n}(4n\alpha)}{\widehat{D}_{n}(\infty)}.
\end{align}

Denoting by $\widehat{\mathbb{P}}(n,\ga,t)$ the probability that the largest eigenvalue of $n\times n$ Hermitian matrices is $\leq t$ in the LUE with the normal Laguerre weight $x^{\ga}\mathrm{e}^{-x}$, we have (see \cite{Lyu2017})
\be\label{pr}
\widehat{\mathbb{P}}(n,\ga,t)=\frac{\widehat{D}_{n}(t)}{\widehat{D}_{n}(\infty)}.
\ee
Note that $\widehat{D}_{n}(\infty)$ has the following closed-form expression \cite[p.321 (17.6.5)]{Mehta},
\begin{align}
\widehat{D}_{n}(\infty)=&\frac{1}{n!}\prod_{j=1}^{n}\Gamma(j+1)\Gamma(j+\ga)\nonumber\\
=&\frac{G(n+1)G(n+\ga+1)}{G(\ga+1)},\label{dni}
\end{align}
where $G(\cdot)$ is the Barnes $G$-function which satisfies the relation
$$
G(z+1)=\Gamma(z)G(z),\qquad\qquad G(1):=1.
$$
See \cite{Barnes,Voros,Choi} for more properties of this function.

A combination of \eqref{probDn1} and \eqref{pr} gives us a connection between the largest eigenvalue distribution of LUE with the weight $x^{\gamma}\mathrm{e}^{-4nx}$ and the weight $x^{\ga}\mathrm{e}^{-x}$:
$$
\mathbb{P}(n,\gamma,\alpha)=\widehat{\mathbb{P}}(n,\ga,4n\alpha).
$$
Therefore, to study $\mathbb{P}(n,\gamma,\alpha)$, we first turn our attention to $\widehat{\mathbb{P}}(n,\ga,t)$.

It is well known that the gap probability $\widehat{\mathbb{P}}(n,\ga,t)$ that the interval $(t,\infty)$ contains no eigenvalues, can be expressed as a Fredholm determinant \cite[p.109 (5.42)]{Deift1999}, namely,
$$
\widehat{\mathbb{P}}(n,\ga,t)=\det\left(I-K_{n}\mathlarger {\chi}_{(t,\infty)}\right),
$$
where $\mathlarger{\mathlarger {\chi}}_{(t,\infty)}(\cdot)$ is the characteristic function of the interval $(t,\infty)$ and the integral operator $K_{n}\mathlarger{\mathlarger {\chi}}_{(t,\infty)}$  has kernel $K_n(x,y)\mathlarger{\mathlarger {\chi}}_{(t,\infty)}(y)$, with $K_n(x,y)$ given by the Christoffel-Darboux formula \cite{Szego},
\begin{align*}
K_{n}(x,y)=&\sum_{j=0}^{n-1}\varphi_j(x)\varphi_j(y)\\
=&\sqrt{n(n+\ga)}\:\frac{\varphi_{n-1}(x)\varphi_{n}(y)-\varphi_{n-1}(y)\varphi_{n}(x)}{x-y}.
\end{align*}
Here $\{\varphi_j(x)\}_{j=0}^{\infty}$ are obtained by orthonormalizing the sequence $\{x^jx^{\gamma/2}\mathrm{e}^{-x/2}\}_{j=0}^{\infty}$ over $[0,\infty)$, and \[\varphi_j(x)=\sqrt{\frac{\Gamma(j+1)}{\Gamma(j+\ga+1)}}x^{\frac{\ga}{2}}\mathrm{e}^{-\frac{x}{2}}L_{j}^{(\ga)}(x),\] with $L_{j}^{(\ga)}(x)$ denoting the Laguerre polynomial of degree $j$.

The kernel $K_{n}(x,y)$ tends to the Airy kernel at the soft edge \cite{Forrester}, i.e.,
$$
\lim_{n\rightarrow\infty}2^{\frac{4}{3}}n^{\frac{1}{3}}K_{n}\left(4n+2\ga+2+2^{\frac{4}{3}}n^{\frac{1}{3}}u,4n+2\ga+2+2^{\frac{4}{3}}n^{\frac{1}{3}}v\right)
=K_{\mathrm{Airy}}(u,v),
$$
where $K_{\mathrm{Airy}}(u,v)$ is the Airy kernel defined by
$$
K_{\mathrm{Airy}}(u,v):=\frac{\mathrm{Ai}(u)\mathrm{Ai}'(v)-\mathrm{Ai}(v)\mathrm{Ai}'(u)}{u-v}.
$$
Here $\mathrm{Ai}(\cdot)$ is the Airy function of the first kind \cite{Lebedev}. See also \cite{TW1,Min2020} on the study of the Airy kernel. Tracy and Widom \cite{TW1} showed that $\widehat{\mathbb{P}}(n,\ga,t)$ can be expressed in terms of a Painlev\'{e} II transcendent at the soft edge.

At the hard edge, $K_{n}(x,y)$ tends to the Bessel kernel \cite{Forrester}, that is,
$$
\lim_{n\rightarrow\infty}\frac{1}{4n}K_{n}\left(\frac{u}{4n},\frac{v}{4n}\right)
=K_{\mathrm{Bessel}}(u,v),
$$
where
\[K_{\mathrm{Bessel}}(u,v):=\frac{J_{\gamma}(\sqrt{u})\sqrt{v}J_{\gamma}'(\sqrt{v})-\sqrt{u}J_{\gamma}'(\sqrt{u})J_{\gamma}(\sqrt{v})}{2(u-v)}.\]
Tracy and Widom \cite{TW2} proved that the log-derivative of $\widehat{\mathbb{P}}(n,\ga,t)$ satisfies a particular Painlev\'{e} III equation when $t$ approaches the hard edge.

The level density of the LUE with the weight $x^{\ga}\mathrm{e}^{-x}$ is given by \cite[p.356 (19.1.11)]{Mehta} \[\rho(x)=\frac{1}{2\pi}\sqrt{\frac{4n-x}{x}},\qquad 0<x<4n,\]
which is an example of the Mar\v{c}enko-Pastur law \cite{MP}. Hence, in \cite{Deift} and also in this paper, the scaled Laguerre weight with $n$ appearing in the exponent is considered, in order to make the equilibrium density of the eigenvalues supported on $(0,1)$ instead of $(0,4n)$.

For finite $n$, Tracy and Widom \cite {TW3} established a particular Painlev\'{e} V equation satisfied by the log-derivative of $\widehat{\mathbb{P}}(n,\ga,t)$. Adler and van Moerbeke \cite{Adler} derived the same results via differential operators.
By using the ladder operator approach of orthogonal polynomials, Basor and Chen \cite{Basor2009} investigated the Hankel determinant generated by the Laguerre weight with a jump, which includes $\widehat{D}_{n}(t)$ as a special case, and a Painlev\'{e} V equation shows up as is expected.
Based on their results, Lyu and Chen \cite{Lyu2017} considered the asymptotic behavior of $x^{\gamma/2}\mathrm{e}^{-x/2}P_j(x)$ at the soft edge, with $P_j(x),\;j=0,1,\ldots$ denoting the monic polynomials orthogonal with respect to $x^{\gamma}\mathrm{e}^{-x}$ on $[0,t]$.
We mention here that the ladder operator method is effective and straightforward in the finite dimensional analysis of problems in unitary ensembles, see for example, the gap probability \cite{Basor2012,Lyu2019,LyuChenFan,Min2018} and the partition function for weights with discontinuities or singularities \cite{Chen2010,ChenZhang,Min2019a,MLC}.

In the present paper, in order to derive the asymptotic formula for $\mathbb{P}(n,\gamma,\alpha)$ at the soft edge, we proceed from two aspects. On one hand, we first derive a large $n$ asymptotic expansion for  $\frac{d}{d\alpha}\ln\mathbb{P}(n,\gamma,\alpha)$, by using differential equations for finite $n$ from \cite{Basor2009}. Then we integrate the expansion from $\alpha_0$ to $\alpha$ with arbitrary $\alpha_0<\alpha$ to obtain an asymptotic formula for $\ln\mathbb{P}(n,\gamma,\alpha)-\ln\mathbb{P}(n,\gamma,\alpha_0)$. On the other hand, we make use of the definition of $D_n(\alpha)$, i.e. the multiple integral, to get an approximate expression for $\ln\mathbb{P}(n,\gamma,\alpha_0)$ when $\al_0$ is close to 0. Taking the sum of these two asymptotic expansions together, and by sending $\alpha_0$ to 0, we come to the asymptotics for $\ln\mathbb{P}(n,\gamma,\alpha)$ in large $n$ with $\alpha$ ranging from 0 to the soft edge, where a term which is independent of $\alpha$ and tends to $0$ as $n\rightarrow\infty$ is included. Finally, by setting $\alpha=1-\frac{s}{(2n)^{2/3}}$ and sending $n$ to $\infty$, we obtain the asymptotic formula of $\ln\mathbb{P}(n,\gamma,\alpha)$ at the soft edge for large $s$,
$$
\lim_{n\rightarrow\infty}\ln\mathbb{P}\left(n,\gamma,1-\frac{s}{(2n)^{2/3}}\right)=-\frac{s^3}{12}-\frac{1}{8}\ln s+\frac{1}{24}\ln 2+\zeta'(-1)+O(s^{-3}),
$$
where the celebrated Tracy-Widom constant \cite{TW1} shows up.

The above method is motivated by Deift, Its and Krasovsky \cite{Deift} where they studied the special case $\mathbb{P}(n,0,\alpha)$, namely the largest eigenvalue distribution on $[0,\alpha]$ of LUE with the weight $\mathrm{e}^{-4nx}$. They used the Riemann-Hilbert approach to get the asymptotic expansion for $\frac{d}{d\alpha}\ln\mathbb{P}(n,0,\alpha)$, while in this paper, as is mentioned above, we relate the weight $x^{\gamma}\mathrm{e}^{-4nx}$ to $x^{\ga}\mathrm{e}^{-x}$ and make use of the established results for the latter weight from \cite{Basor2009}. The Riemann-Hilbert method \cite{Deift1999} is a very powerful tool to investigate the asymptotic behavior of many unitary ensembles. See, for instance, the gap probability problem \cite{DaiXuZhang2018,XuZhao2019}, correlation kernel \cite{ChenChenFan,XuDaiZhao}, partition functions \cite{DaiXuZhang2019, BMM, ACM}, Hankel determinants and orthogonal polynomials \cite{Bogatskiy,Charlier,ChenChenFan2019,Xu2011}.

This paper is organized as follows. In Sec. 2, we present some important results from \cite{Basor2009} which are related to the largest eigenvalue distribution of LUE with the weight $x^{\ga}\mathrm{e}^{-x}$. Sec. 3 is devoted to the derivation of the asymptotic formula for $\frac{d}{d\alpha}\ln\mathbb{P}(n,\gamma,\alpha)$ when $n$ is large. Our main results are developed in Sec. 4.

\section{Preliminaries}
In this section, we show some important results of Basor and Chen \cite{Basor2009}, which are crucial for the analysis of the asymptotic behavior of the largest eigenvalue distribution in the LUE with the scaled Laguerre weight.

It is well known that the multiple integral $\widehat{D}_n(t)$ defined by \eqref{dnt} can be written as the determinant of a Hankel matrix and also as the product of the square of the $L^2$-norms of the corresponding monic orthogonal polynomials \cite[p.16-19]{Ismail}, namely
\begin{align*}
\widehat{D}_n(t):=&\frac{1}{n!}\int_{[0,t]^{n}}\prod_{1\leq i<j\leq n}(x_i-x_j)^2\prod_{k=1}^{n}x_k^{\ga}\mathrm{e}^{-x_k}dx_{k}\\
=&\det\left(\int_{0}^{t}x^{i+j}x^{\ga}\mathrm{e}^{-x}dx\right)_{i,j=0}^{n-1}\\
=&\prod_{j=0}^{n-1}h_j(t),
\end{align*}
where
\be\label{or}
h_j(t)\delta_{jk}:=\int_{0}^{t}P_j(x,t)P_k(x,t)x^{\ga}\mathrm{e}^{-x}dx,
\ee
and $\delta_{jk}$ is the Kronecker delta function. Here $P_n(x,t),\; n=0,1,2,\ldots,$ are monic polynomials of degree $n$ defined by
\be\label{expan}
P_{n}(x,t)=x^{n}+\mathrm{p}(n,t)x^{n-1}+\cdots+P_{n}(0,t).
\ee
Note that, in the following discussions, $n$ stands for any nonnegative integer instead of the dimension of the Hermitian matrices.

The orthogonality (\ref{or}) implies the following three-term recurrence relation \cite{Chihara,Szego}:
\be\label{rr}
xP_{n}(x,t)=P_{n+1}(x,t)+\alpha_{n}(t)P_{n}(x,t)+\beta_{n}(t)P_{n-1}(x,t),
\ee
subject to the initial conditions
$$
P_{0}(x,t):=1,\qquad\qquad\beta_{0}(t)P_{-1}(x,t):=0.
$$

As an easy consequence of (\ref{or})--(\ref{rr}), we have
$$
\alpha_{n}(t)=\mathrm{p}(n,t)-\mathrm{p}(n+1,t),
$$
$$
\beta_{n}(t)=\frac{h_{n}(t)}{h_{n-1}(t)}.
$$
In addition, $\alpha_{n}(t)$ and $\beta_{n}(t)$ admit the following integral representations,
$$
\alpha_{n}(t)=\frac{1}{h_{n}(t)}\int_{0}^{t}x P_{n}^{2}(x,t)x^{\ga}\mathrm{e}^{-x}dx,
$$
$$
\beta_{n}(t)=\frac{1}{h_{n-1}(t)}\int_{0}^{t}x P_{n}(x,t)P_{n-1}(x,t)x^{\ga}\mathrm{e}^{-x}dx.
$$
From the recurrence relation (\ref{rr}), one can derive the famous Christoffel-Darboux formula \cite{Szego},
$$
\sum_{k=0}^{n-1}\frac{P_{k}(x,t)P_{k}(y,t)}{h_{k}(t)}=\frac{P_{n}(x,t)P_{n-1}(y,t)-P_{n}(y,t)P_{n-1}(x,t)}{h_{n-1}(t)(x-y)}.
$$

For convenience, we will not display the $t$ dependence of relevant quantities unless it is required in the following discussions.

Basor and Chen \cite{Basor2009} studied the Hankel determinant generated by the discontinuous Laguerre weight $x^{\gamma}\mathrm{e}^{-x}(A+B\theta(x-t))$, where $\theta(\cdot)$ is the Heaviside step function. We observe that the special case where $A=1$ and $B=-1$ corresponds to $\widehat{D}_n(t)$.

It is proved in \cite{Basor2009} that the monic orthogonal polynomials defined by \eqref{or} satisfy the lowering operator equation
$$
\left(\frac{d}{dz}+B_{n}(z)\right)P_{n}(z)=\beta_{n}A_{n}(z)P_{n-1}(z),
$$
and the raising operator equation
$$
\left(\frac{d}{dz}-B_{n}(z)-\mathrm{v}'(z)\right)P_{n-1}(z)=-A_{n-1}(z)P_{n}(z),
$$
where $\mathrm{v}(z):=-\ln(z^{\gamma}\mathrm{e}^{-z})=z-\ga\ln z$, and
$$
A_{n}(z):=\frac{R_n(t)}{z-t}+\frac{1}{h_{n}(t)}\int_{0}^{t}\frac{\mathrm{v}'(z)-\mathrm{v}'(y)}{z-y}P_{n}^{2}(y)y^{\ga}\mathrm{e}^{-y}dy,
$$
$$
B_{n}(z):=\frac{r_n(t)}{z-t}+\frac{1}{h_{n-1}(t)}\int_{0}^{t}\frac{\mathrm{v}'(z)-\mathrm{v}'(y)}{z-y}P_{n}(y)P_{n-1}(y)y^{\ga}\mathrm{e}^{-y}dy.
$$
Here the auxiliary quantities $R_n(t)$ and $r_n(t)$ are defined by
$$
R_n(t):=-\frac{t^{\ga}\mathrm{e}^{-t}}{h_n(t)}P_n^2(t,t),
$$
$$
r_n(t):=-\frac{t^{\ga}\mathrm{e}^{-t}}{h_{n-1}(t)}P_n(t,t)P_{n-1}(t,t),
$$
and $P_n(t,t):=P_n(z,t)|_{z=t}$.

From the definitions of $A_n(z)$ and $B_n(z)$, Basor and Chen \cite{Basor2009} derived two identities valid for $z\in\mathbb{C}\cup\{\infty\}$:
\be
B_{n+1}(z)+B_{n}(z)=(z-\alpha_{n})A_{n}(z)-\mathrm{v}'(z), \tag{$S_{1}$}
\ee
\be
1+(z-\alpha_{n})(B_{n+1}(z)-B_{n}(z))=\beta_{n+1}A_{n+1}(z)-\beta_{n}A_{n-1}(z). \tag{$S_{2}$}
\ee
A combination of $(S_1)$ and $(S_2)$ produces
\be
B_{n}^{2}(z)+\mathrm{v}'(z)B_{n}(z)+\sum_{j=0}^{n-1}A_{j}(z)=\beta_{n}A_{n}(z)A_{n-1}(z). \tag{$S_{2}'$}
\ee

Computing $A_n(z)$ and $B_n(z)$ by using their definitions and substituting the resulting expressions into the compatibility conditions $(S_1)$, $(S_2)$ and $(S_2')$, Basor and Chen \cite{Basor2009} obtained the following results.
\begin{proposition} $A_n(z)$ and $B_n(z)$ are given by
$$
A_{n}(z)=\frac{R_{n}(t)}{z-t}+\frac{1-R_{n}(t)}{z},
$$
$$
B_{n}(z)=\frac{r_{n}(t)}{z-t}-\frac{n+r_{n}(t)}{z}.
$$
\end{proposition}

%
\begin{proposition}\label{p1}
The quantity
$$S_{n}(t):=1-\frac{1}{R_n(t)},$$
satisfies the second-order differential equation
\be\label{pv}
S_{n}''=\frac{(3S_{n}-1)(S_{n}')^2}{2S_{n}(S_{n}-1)}-\frac{S_{n}'}{t}-\frac{\ga^{2}}{2}\frac{(S_{n}-1)^2}{t^2S_{n}}+\frac{(2n+1+\ga)S_{n}}{t}-\frac{S_{n}(S_{n}+1)}{2(S_{n}-1)},
\ee
which is a particular Painlev\'{e} V equation, $P_{V}\left(0,-\frac{\ga^2}{2},2n+1+\ga,-\frac{1}{2}\right)$, following the convention of \cite{Gromak}.
\end{proposition}
Let
$$
\s_n(t):=t\frac{d}{dt}\ln \widehat{D}_n(t).
$$
Then, in view of \eqref{pr}, we have
\begin{align}\label{sn}
\s_n(t)=t\frac{d}{dt}\ln\widehat{\mathbb{P}}(n,\ga,t).
\end{align}
Recall that $\widehat{\mathbb{P}}(n,\ga,t)$ represents the largest eigenvalue distribution function on $[0,t]$ of LUE with the weight $x^{\ga}\mathrm{e}^{-x}$. The following results come from \cite{Basor2009} and \cite{Lyu2017}.

\begin{proposition}\label{p2}
The quantity $\s_n(t)$ satisfies the Jimbo-Miwa-Okamoto $\s$-form of Painlev\'{e} V \cite{Jimbo},
$$
(t\s_n'')^2=4(\s_n')^2\left(\s_n-n(n+\ga)-t\s_n'\right)+\left((2n+\ga-t)\s_n'+\s_n\right)^2.
$$
In addition, $\sigma_n(t)$ is expressed in terms of $S_n(t)$ by
\bea\label{ss}
\s_n(t)
&=&-\frac{\ga^2}{4S_n}+\frac{t(4n+2\ga-t)}{4(S_n-1)}-\frac{t^2}{4(S_n-1)^2}+\frac{t^2(S_n')^2}{4S_n(S_n-1)^2}.
\eea
\end{proposition}

In the next section, we will make use of the above results to study $\mathbb{P}(n,\gamma,\alpha)$, that is, the probability that the largest eigenvalue is $\leq\alpha$ of the LUE with the weight $x^{\gamma}\mathrm{e}^{-4nx}$.

\section{Logarithmic Derivative of the Largest Eigenvalue Distribution Function}
We consider the LUE defined by the scaled Laguerre weight,
$$
w(x;\gamma,n)=x^{\gamma}\mathrm{e}^{-4nx},\qquad x\in[0,\infty),\quad\gamma>-1.
$$
As is shown in the introduction, the probability that the largest eigenvalue in this LUE is $\leq\alpha$ is equal to the probability that the largest eigenvalue is $\leq 4n\alpha$ in the LUE with the weight $x^{\ga}\mathrm{e}^{-x}$, i.e.,
%
\be\label{re}
\mathbb{P}(n,\gamma,\alpha)=\widehat{\mathbb{P}}(n,\gamma,4n\alpha).
\ee

According to the results in the previous section with $t=4n\alpha$, we come to the following result.

\begin{lemma} As $n\rightarrow\infty$, $\frac{d}{d\al}\ln\mathbb{P}(n,\gamma,\alpha)$ has the following asymptotic expansion,
\bea\label{phi}
\frac{d}{d\al}\ln\mathbb{P}(n,\gamma,\alpha)&=&\frac{(1-\al)^2}{\al}n^2+\frac{\ga(1-\al)}{\al}n+\frac{\al+2\ga^2(1-\al)}{4(1-\al^2)}
-\frac{\ga\left(\al+\ga^2(1-\al)^2\right)}{4n(1-\al^2)^2}\nonumber\\[5pt]
&&+O\left(\frac{1}{(1-\al)^4n^2}\right).
\eea
\end{lemma}
\begin{proof}
Let
$$
t=4n\alpha,
$$
and denote
\be\label{t1}
F_n(\al):=S_n(t)=S_n(4n\al).
\ee
Then equation (\ref{pv}) becomes
\be\label{fn}
F_n''=\frac{(3F_n-1)(F_n')^2}{2F_n(F_n-1)}-\frac{F_n'}{\al}+\frac{4n(2n+1+\ga)F_n}{\al}-\frac{\ga^2(F_n-1)^2}{2\al^2F_n}-\frac{8n^2F_n(F_n+1)}{F_n-1}.
\ee
In order to obtain the asymptotic formula of $F_n(\al)$, we disregard the derivative terms in this equation and have
$$
\frac{4n(2n+1+\ga)\tilde{F}_n}{\al}-\frac{\ga^2(\tilde{F}_n-1)^2}{2\al^2\tilde{F}_n}-\frac{8n^2\tilde{F}_n(\tilde{F}_n+1)}{\tilde{F}_n-1}=0,
$$
which is actually a cubic equation for $\tilde{F}_n(\al)$, i.e.,
$$
\left(16n^2\al(1-\al)+8n\al(1+\ga)-\ga^2\right)\tilde{F}_n^3-\left(16n^2\al(1+\al)+8n\al(1+\ga)-3\ga^2\right)\tilde{F}_n^2-3\ga^2\tilde{F}_n+\ga^2=0.
$$
It has only one real solution which has the following large $n$ expansion,
$$
\tilde{F}_n(\al)=\frac{1+\al}{1-\al}-\frac{\al(1+\ga)}{n(1-\al)^2}+\frac{\al(1+\al)^2(1+2\ga)+\al(1+3\al)\ga^2}{2n^2(1-\al)^3(1+\al)^2}+O\left(\frac{1}{n^3(1-\al)^4}\right).
$$
Hence we suppose that $F_n(\al)$ has the following series expansion,
$$
F_n(\al)=\sum_{i=0}^{\infty}a_i(\al)n^{-i},\qquad n\rightarrow\infty.
$$
Substituting the above expression into equation (\ref{fn}), and comparing the coefficients of identical powers of $n$ on both sides, we obtain $a_i(\al), i=0, 1, 2, \ldots$ one by one. This leads to the expansion
\bea\label{fne}
F_n(\al)&=&\frac{1+\al}{1-\al}-\frac{\al(1+\ga)}{n(1-\al)^2}+\frac{\al+\al^2-\al^4+2\al(1-\al)(1+\al)^2\ga+\al(1-\al)(1+3\al)\ga^2}{2n^2(1-\al)^4(1+\al)^2}\nonumber\\[10pt]
&&+O\left(\frac{1}{n^3(1-\al)^5}\right).
\eea

Furthermore, with the relation $t=4n\alpha$, it follows from (\ref{sn}) and (\ref{re}) that
\begin{align}
\sigma_{n}(t)=&\sigma_{n}(4n\al)=\al\frac{d}{d\al}\ln\widehat{\mathbb{P}}(n,\ga,4n\al)\nonumber\\
=&\al\frac{d}{d\al}\ln\mathbb{P}(n,\gamma,\alpha).\nonumber
\end{align}
Hence, by using (\ref{ss}) and in view of (\ref{t1}), we are able to express $\frac{d}{d\al}\ln\mathbb{P}(n,\gamma,\alpha)$ in terms of $F_n(\al)$,
\be\label{pf}
\frac{d}{d\al}\ln\mathbb{P}(n,\gamma,\alpha)=-\frac{\ga^2}{4\al F_n}+\frac{2n(2n(1-\al)+\ga)}{F_n-1}-\frac{4n^2\al}{(F_n-1)^2}+\frac{\al(F_n')^2}{4F_n(F_n-1)^2}.
\ee
Substituting (\ref{fne}) into (\ref{pf}), we arrive at \eqref{phi}.
\end{proof}
\begin{remark}
 When $\ga=0$, our formula (\ref{phi}) coincides with the expression (152) of Deift et al. \cite{Deift}.
\end{remark}

In order to obtain the asymptotic formula of the largest eigenvalue distribution function $\mathbb{P}(n,\gamma,\alpha)$ as $n\rightarrow\infty$, we can integrate identity (\ref{phi}) from $\al_0$ to any $\al$, where $\alpha_0$ is close to 0 and $0<\al_0<\al\leq\frac{1}{4n}(4n-2^{4/3}n^{1/3}s_0)=1-\frac{s_0}{(2n)^{2/3}}$ with finite $s_0>0$. See \cite{Perret, Lyu2017} for the study on the soft edge scaling limit of LUE. So we need to know the asymptotics of $\mathbb{P}(n,\gamma,\alpha)$ when $\al$ tends to 0. We will analyze it in the next section following the method in \cite{Deift}.

\section{Asymptotic Behavior of the Largest Eigenvalue Distribution Function}
Returning to our problem, we recall that the probability that all the eigenvalues lie in $[0,\alpha]$ is given by
$$
\mathbb{P}(n,\gamma,\alpha)=\frac{D_{n}(\alpha)}{D_{n}(\infty)},
$$
where
$$
D_{n}(\alpha)=\frac{1}{n!}\int_{[0,\alpha]^{n}}\prod_{1\leq i<j\leq n}(x_i-x_j)^2\prod_{k=1}^{n}x_{k}^{\ga}\mathrm{e}^{-4nx_{k}}dx_{k},
$$
and
$$
D_{n}(\infty)=\frac{1}{n!}\int_{[0,\infty)^{n}}\prod_{1\leq i<j\leq n}(x_i-x_j)^2\prod_{k=1}^{n}x_{k}^{\ga}\mathrm{e}^{-4nx_{k}}dx_{k}.
$$
By changing variables $x_{\ell}=\al t_{\ell}, {\ell}=1,2,\ldots,n$, we find
$$
D_{n}(\alpha)=\al^{n(n+\ga)}\frac{1}{n!}\int_{[0,1]^{n}}\prod_{1\leq i<j\leq n}(t_i-t_j)^2\prod_{k=1}^{n}t_{k}^{\ga}\mathrm{e}^{-4n\al t_{k}}dt_{k}.
$$
For fixed $n$ and as $\al\rightarrow 0$, we have
$$
\mathrm{e}^{-4n\al t_{k}}=1-4n\al t_{k}+O(\al^2),
$$
so that
\bea
D_{n}(\alpha)&=&\al^{n(n+\ga)}\frac{1}{n!}\int_{[0,1]^{n}}\prod_{1\leq i<j\leq n}(t_i-t_j)^2\prod_{k=1}^{n}t_{k}^{\ga}(1-4n\al t_{k}+O(\al^2))dt_{k}\nonumber\\
&=&\al^{n(n+\ga)}A_n(\gamma)(1+o_n(\al)),\nonumber
\eea
where $o_n(\al)\rightarrow 0$ as $\al\rightarrow 0$ for fixed $n$, and
$$
A_n(\gamma):=\frac{1}{n!}\int_{[0,1]^{n}}\prod_{1\leq i<j\leq n}(t_i-t_j)^2\prod_{k=1}^{n}t_{k}^{\ga}dt_{k}.
$$
Hence, we find as $\al\rightarrow 0$,
\bea\label{pn}
\ln\mathbb{P}(n,\gamma,\alpha)&=&\ln D_{n}(\alpha)-\ln D_{n}(\infty)\nonumber\\
&=&n(n+\ga)\ln\al+\ln A_n(\gamma)-\ln D_{n}(\infty)+o_n(\al).
\eea

According to identity (17.1.3) in \cite{Mehta}, we have
\bea
A_n(\gamma)
&=&\frac{1}{n!}\prod_{j=0}^{n-1}\frac{\Gamma(j+1)\Gamma(j+2)\Gamma(j+\ga+1)}{\Gamma(j+n+\ga+1)}\nonumber\\
&=&\frac{G^{2}(n+1)G^2(n+\ga+1)}{G(\ga+1)G(2n+\ga+1)},\nonumber
\eea
where $G(\cdot)$ is the Barnes $G$-function. Now we look at $D_{n}(\infty)$, i.e.
$$
D_{n}(\infty)=\frac{1}{n!}\int_{[0,\infty)^{n}}\prod_{1\leq i<j\leq n}(x_i-x_j)^2\prod_{k=1}^{n}x_{k}^{\ga}\mathrm{e}^{-4nx_{k}}dx_{k}.
$$
By changing variables $4nx_{\ell}=y_{\ell}, \ell=1,2,\ldots,n$, we get
\bea
D_{n}(\infty)&=&(4n)^{-n(n+\ga)}\cdot\frac{1}{n!}\int_{[0,\infty)^{n}}\prod_{1\leq i<j\leq n}(y_i-y_j)^2\prod_{k=1}^{n}y_{k}^{\ga}\mathrm{e}^{-y_{k}}dy_{k}\nonumber\\[5pt]
&=&(4n)^{-n(n+\ga)}\widehat{D}_{n}(\infty)\nonumber\\[5pt]
&=&(4n)^{-n(n+\ga)}\frac{G(n+1)G(n+\ga+1)}{G(\ga+1)},\nonumber
\eea
where we have used (\ref{dni}). Substituting the above expressions for $A_n(\gamma)$ and $D_{n}(\infty)$ into (\ref{pn}), we arrive at, as $\al\rightarrow 0$,
$$
\ln\mathbb{P}(n,\gamma,\alpha)=n(n+\ga)\ln(4n\al)+\ln G(n+1)+\ln G(n+\ga+1)-\ln G(2n+\ga+1)+o_n(\al).
$$

By using the asymptotic formula of Barnes $G$-function (see, for example, formula (A.6) in \cite{Voros}), i.e.,
\be\label{bg}
\ln G(z+1)=z^2\left(\frac{\ln z}{2}-\frac{3}{4}\right)+\frac{z}{2}\ln (2\pi)-\frac{\ln z}{12}+\zeta'(-1)+O(z^{-1}),\qquad z\rightarrow\infty,
\ee
where $\zeta(\cdot)$ is the Riemann zeta function, we obtain
\bea\label{pn1}
\ln\mathbb{P}(n,\gamma,\alpha)&=&\left(\frac{3}{2}n^2+n\ga-\frac{1}{12}\right)\ln n+\left(\frac{(n+\ga)^2}{2}-\frac{1}{12}\right)\ln(n+\ga)\nonumber\\[10pt]
&&-\left(\frac{(2n+\ga)^2}{2}-\frac{1}{12}\right)\ln(2n+\ga)+n(n+\ga)\left(\frac{3}{2}+\ln (4\al)\right)\nonumber\\[10pt]
&&+\zeta'(-1)+\tilde{\delta}_{n}(\ga)+o_n(\al),
\eea
where $\tilde{\delta}_{n}(\ga)$ depends only on $n$ and $\ga$, and $\tilde{\delta}_{n}(\ga)\rightarrow 0$ as $n\rightarrow\infty$ for any given $\ga$.

\begin{remark}
When $\ga=0$, formula (\ref{pn1}) is consistent with formula (27) of Deift et al. \cite{Deift}.
\end{remark}

To continue, we integrate identity (\ref{phi}) from $\al_0$ to any $\al$ with $0<\al_0<\al\leq1-\frac{s_0}{(2n)^{2/3}},\; s_0>0$, and find
\bea\label{pn2}
&&\ln\mathbb{P}(n,\ga,\al)-\ln\mathbb{P}(n,\ga,\al_0)\nonumber\\[5pt]
&=&n^2\left(\ln\frac{\al}{\al_0}+\frac{\al^2-\al_0^2}{2}-2(\al-\al_0)\right)
+n\ga\left(\ln\frac{\al}{\al_0}-(\al-\al_0)\right)\nonumber\\[10pt]
&&+\frac{1}{8}\left((4\ga^2-1)\ln\frac{1+\al}{1+\al_0}-\ln\frac{1-\al}{1-\al_0}\right)+\frac{\ga\left(2(1-\al)\ga^2-1\right)}{8n(1-\al^2)}\nonumber\\[10pt]
&&-\frac{\ga\left(2(1-\al_0)\ga^2-1\right)}{8n(1-\al_0^2)}+O\left(\frac{1}{n^2(1-\al)^3}\right)-O\left(\frac{1}{n^2(1-\al_0)^3}\right).
\eea
Substituting formula (\ref{pn1}) for $\ln\mathbb{P}(n,\ga,\al_0)$ into (\ref{pn2}) and taking the limit $\al_0\rightarrow 0$, we establish the following theorem.
\begin{theorem}
For any $0<\al\leq1-\frac{s_0}{(2n)^{2/3}},\; s_0>0$, we have as $n\rightarrow\infty$,
\bea\label{pn3}
\ln\mathbb{P}(n,\ga,\al)&=&n^2\left(\frac{3}{2}-2\al+\frac{\al^2}{2}+\ln(4\al)\right)+n\ga\left(\frac{3}{2}-\al+\ln(4\al)\right)
+\left(\frac{3}{2}n^2+n\ga-\frac{1}{12}\right)\ln n\nonumber\\
&&+\left(\frac{(n+\ga)^2}{2}-\frac{1}{12}\right)\ln(n+\ga)-\left(\frac{(2n+\ga)^2}{2}-\frac{1}{12}\right)\ln(2n+\ga)\nonumber\\
&&+\frac{1}{8}\left((4\ga^2-1)\ln(1+\al)-\ln(1-\al)\right)+\zeta'(-1)+\frac{\ga\left(2(1-\al)\ga^2-1\right)}{8n(1-\al^2)}\nonumber\\
&&+O\left(\frac{1}{n^2(1-\al)^3}\right)+\delta_{n}(\ga),
\eea
where $\delta_{n}(\ga)$ depends only on $n$ and $\ga$, and $\delta_{n}(\ga)\rightarrow 0$ as $n\rightarrow\infty$ for any given $\ga$.
\end{theorem}
\begin{remark}
From the asymptotic formula (\ref{bg}) of the Barnes $G$-function, we can show that $\delta_{n}(\ga)=O(\frac{1}{n})\: (n\rightarrow\infty)$ for any given $\ga$. In addition, if $\ga=0$, then (\ref{pn3}) becomes
\bea
\ln\mathbb{P}(n,0,\al)&=&n^2\left(\frac{3}{2}+\ln\al-2\al+\frac{\al^2}{2}\right)-\frac{1}{12}\ln n-\frac{1}{8}\ln(1-\al^2)\nonumber\\
&&+\frac{1}{12}\ln 2+\zeta'(-1)+O\left(\frac{1}{n^2(1-\al)^3}\right)+\delta_{n},\nonumber
\eea
where $\delta_{n}$ depends only on $n$, and $\delta_{n}\rightarrow 0$ as $n\rightarrow\infty$.
This agrees with (162) of Deift et al. \cite{Deift}.
\end{remark}

In the end, we give the asymptotic formula of the Fredholm determinant
$$
\det(I-K_{\mathrm{Airy}}),
$$
where $K_{\mathrm{Airy}}$ is the integral operator with the Airy kernel
$$
K_{\mathrm{Airy}}(x,y)=\frac{\mathrm{Ai}(x)\mathrm{Ai}'(y)-\mathrm{Ai}(y)\mathrm{Ai}'(x)}{x-y}
$$
acting on $L^{2}(-s,\infty)$.

For any $s>s_0$ and sufficiently large $n$, we set
$$
\al=1-\frac{s}{(2n)^{2/3}}.
$$
Substituting it into (\ref{pn3}) and taking the limit $n\rightarrow\infty$, the r.h.s. of (\ref{pn3}) becomes
$$
-\frac{s^3}{12}-\frac{1}{8}\ln s+\frac{1}{24}\ln 2+\zeta'(-1)+O(s^{-3}),
$$
and the l.h.s. of (\ref{pn3}) approaches $\ln\det(I-K_{\mathrm{Airy}})$. Therefore, we establish the following asymptotic formula of the Airy determinant as $s\rightarrow+\infty$,
\be\label{airy}
\ln\det(I-K_{\mathrm{Airy}})=-\frac{s^3}{12}-\frac{1}{8}\ln s+\frac{1}{24}\ln 2+\zeta'(-1)+O(s^{-3}).
\ee

\begin{remark}
The constant term in (\ref{airy}), i.e. $\frac{1}{24}\ln 2+\zeta'(-1)$,
was conjectured by Tracy and Widom \cite{TW1}, and proved by Deift et al. \cite{Deift} where (\ref{airy}) was also derived but with the order term $O(s^{-3/2})$. Our order term $O(s^{-3})$ coincides with formula (1.19) in \cite{TW1}. In addition, Baik et al. \cite{Baik} gave an alternative proof of (\ref{airy}) by using an integral expression of the Tracy-Widom distribution.
\end{remark}

\section*{Acknowledgments}
The work of Shulin Lyu was supported by National Natural Science Foundation of China under grant number 11971492. Chao Min was supported by the Scientific Research Funds of Huaqiao University under grant number 600005-Z17Y0054.
Yang Chen was supported by the Macau Science and Technology Development Fund under grant number FDCT 023/2017/A1 and by the University of Macau under grant number MYRG 2018-00125-FST.


\begin{thebibliography}{}
%
%
\bibitem{Adler}
{M. Adler} and {P. van Moerbeke}, Hermitian, symmetric and symplectic random ensembles: PDEs for the distribution of the spectrum, {Ann. Math.} {\bf 153} ({2001}), {149--189}.
\bibitem{ACM}
M. Atkin, T. Claeys and F. Mezzadri, Random matrix ensembles with singularities and a hierarchy of Painlev\'{e} III equations, Int. Math. Res. Notices {\bf 2016} (2016), 2320--2375.
\bibitem{Baik}
{J. Baik}, {R. Buckingham} and {J. DiFranco}, Asymptotics of Tracy-Widom distributions and the total integral of a Painlev\'{e} II function, {Commun. Math. Phys.} {\bf 280} ({2008}), {463--497}.
\bibitem{Barnes}
{E. W. Barnes}, The theory of the double gamma function, {Phil. Trans. R. Soc. Lond. A} {\bf 196} ({1901}), {265--387}.
\bibitem{Basor2009}
{E. Basor} and {Y. Chen},  Painlev\'{e} V and the distribution function of a discontinuous linear statistic in the Laguerre unitary ensembles, {J. Phys. A: Math. Theor.} {\bf 42} ({2009}), {035203 (18 pages)}.
\bibitem{Basor2012}
{E. Basor}, {Y. Chen} and {L. Zhang}, PDEs satisfied by extreme eigenvalues distributions of GUE and LUE, {Random Matrices: Theor. Appl.} {\bf 1} ({2012}), {1150003} (21 pages).
\bibitem{Bogatskiy}
{A. Bogatskiy}, {T. Claeys} and {A. Its}, Hankel determinant and orthogonal polynomials for a Gaussian weight with a discontinuity at the edge, {Commun. Math. Phys.} {\bf 347} ({2016}), {127--162}.
\bibitem{BMM}
L. Brightmore, F. Mezzadri and M. Y. Mo, A matrix model with a singular weight and Painlev\'{e} III, Commun. Math. Phys. {\bf 333} (2015), 1317--1364.
\bibitem{Charlier}
C. Charlier and R. Gharakhloo, Asymptotics of Hankel determinants with a Laguerre-type or Jacobi-type potential and Fisher-Hartwig singularities, arXiv: 1902.08162.
\bibitem{ChenChenFan}
M. Chen, Y. Chen and E. Fan, Critical edge behavior in the perturbed Laguerre unitary ensemble and the Painlev\'{e} V transcendent, J. Math. Anal. Appl. {\bf 474} (2019), 572--611.

\bibitem{ChenChenFan2019}
M. Chen, Y. Chen and E. Fan, The Riemann-Hilbert analysis to the Pollaczek-Jacobi type orthogonal polynomials, Stud. Appl. Math. {\bf 143} (2019), 42--80.

\bibitem{Chen2010}
{Y. Chen} and {A. Its}, Painlev\'{e} III and a singular linear statistics in Hermitian random matrix ensembles, I, {J. Approx. Theory} {\bf 162} ({2010}), {270--297}.
\bibitem{ChenZhang}
{Y. Chen} and {L. Zhang}, Painlev\'{e} VI and the unitary Jacobi ensembles, {Stud. Appl. Math.} {\bf 125} ({2010}), {91--112}.
\bibitem{Chihara}
{T. S. Chihara}, {\em An Introduction to Orthogonal Polynomials}, {Dover}, {New York}, {1978}.
\bibitem{Choi}
J. Choi, H. M. Srivastava and V. S. Adamchik, Multiple gamma and related functions, Appl. Math. Comput. {\bf 134} (2003), 515--533.
\bibitem{DaiXuZhang2018}
D. Dai, S.-X. Xu and L. Zhang, Gap probability at the hard edge for random matrix ensembles with pole singularities in the potential, SIAM J. Math. Anal. {\bf 50} (2018), 2233--2279.

\bibitem{DaiXuZhang2019}
D. Dai, S.-X. Xu and L. Zhang, Gaussian unitary ensembles with pole singularities near the soft edge and a system of coupled Painlev\'{e} XXXIV equations, Ann. Henri Poincar\'{e} {\bf 20} (2019), 3313--3364.
\bibitem{Deift1999}
{P. Deift}, {\em Orthogonal Polynomials and Random Matrices: A Riemann-Hilbert Approach}, {Courant Lecture Notes 3}, {New York University}, {Amer. Math. Soc.}, {Providence, RI}, {1999}.
\bibitem{Deift}
{P. Deift}, {A. Its} and {I. Krasovsky},  Asymptotics of the Airy-kernel determinant, {Commun. Math. Phys.} {\bf 278} ({2008}), {643--678}.

\bibitem{Forrester}
{P. J. Forrester}, The spectrum edge of random matrix ensembles, {\em Nucl. Phys. B} {\bf 402} ({1993}), {709--728}.
\bibitem{Gromak}
{V. I. Gromak}, {I. Laine} and {S. Shimomura}, {\em Painlev\'{e} Differential Equations in the Complex Plane}, {Walter de Gruyter}, {Berlin}, {2002}.
\bibitem{Ismail}
{M. E. H. Ismail}, {\em Classical and Quantum Orthogonal Polynomials in One Variable}, {Encyclopedia of Mathematics and its Applications 98}, {Cambridge University Press}, {Cambridge}, {2005}.
\bibitem{Jimbo}
{M. Jimbo} and {T. Miwa},  Monodromy preserving deformation of linear ordinary differential equations with rational coefficients. II, {Physica D} {\bf 2} ({1981}), {407--448}.
\bibitem{Lebedev}
{N. N. Lebedev}, {\em Special Functions and Their Applications}, {Dover Publications}, {New York}, {1972}.
\bibitem{Lyu2017}
{S. Lyu} and {Y. Chen}, The largest eigenvalue distribution of the Laguerre unitary ensemble, {Acta Math. Sci.} {\bf 37B} ({2017}), {439--462}.

\bibitem{Lyu2019}
{S. Lyu} and {Y. Chen}, The Hankel determinant associated with a singularly perturbed Laguerre unitary ensemble, {J. Nonlinear Math. Phys.} {\bf 26} ({2019}), {24--53}.

\bibitem{LyuChenFan}
S. Lyu, Y. Chen and E. Fan, Asymptotic gap probability distributions of the Gaussian unitary ensembles and Jacobi unitary ensembles, Nucl. Phys. B {\bf 926} (2018), 639--670.

\bibitem{MP}
V. A. Mar\v{c}enko and L. A. Pastur, Distributions of eigenvalues for some sets of random matrices, Math. USSR Sbornik {\bf 1} (1967), 457--483.

\bibitem{Mehta}
{M. L. Mehta}, {\em Random Matrices}, {3rd edn.}, {Elsevier}, {New York}, {2004}.
\bibitem{Min2018}
{C. Min} and {Y. Chen}, Gap probability distribution of the Jacobi unitary ensemble: an elementary treatment, from finite $n$ to double scaling, {Stud. Appl. Math.} {\bf 140} ({2018}), {202--220}.
\bibitem{Min2019a}
{C. Min} and {Y. Chen}, Painlev\'{e} transcendents and the Hankel determinants generated by a discontinuous Gaussian weight, {Math. Meth. Appl. Sci.} {\bf 42} ({2019}), {301--321}.
\bibitem{Min2020}
{C. Min} and {Y. Chen}, Linear statistics of random matrix ensembles at the spectrum edge associated with the Airy kernel, {Nucl. Phys. B} {\bf 950} ({2020}), {114836} (34 pages).
\bibitem{MLC}
{C. Min}, {S. Lyu} and {Y. Chen}, Painlev\'{e} III$'$ and the Hankel determinant generated by a singularly perturbed Gaussian weight, {Nucl. Phys. B} {\bf 936} ({2018}), {169--188}.
\bibitem{Perret}
{A. Perret} and {G. Schehr}, Finite $N$ corrections to the limiting distribution of the smallest eigenvalue of Wishart complex matrices, {Random Matrices: Theor. Appl.} {\bf 5} ({2016}), {1650001} (27 pages).
\bibitem{Szego}
{G. Szeg\H{o}}, {\em Orthogonal Polynomials}, {4th edn.}, {AMS Colloquium Publications}, {Vol. 23}, {Providence, RI}, {1975}.
\bibitem{TW1}
{C. A. Tracy} and {H. Widom}, Level-spacing distributions and the Airy kernel, {Commun. Math. Phys.} {\bf 159} ({1994}), {151--174}.
\bibitem{TW2}
{C. A. Tracy} and {H. Widom}, Level spacing distributions and the Bessel kernel, {Commun. Math. Phys.} {\bf 161} ({1994}), {289--309}.
\bibitem{TW3}
{C. A. Tracy} and {H. Widom}, Fredholm determinants, differential equations and matrix models, {Commun. Math. Phys.} {\bf 163} ({1994}), {33--72}.
\bibitem{Voros}
{A. Voros}, Spectral functions, special functions and the Selberg zeta function, {Commun. Math. Phys.} {\bf 110} ({1987}), {439--465}.
\bibitem{XuDaiZhao}
S.-X. Xu, D. Dai and Y.-Q. Zhao, Critical edge behavior and the Bessel to Airy transition in the singularly perturbed Laguerre unitary ensemble, Commun. Math. Phys. {\bf 332} (2014), 1257--1296.
\bibitem{Xu2011}
{S.-X. Xu} and {Y.-Q. Zhao}, Painlev\'{e} XXXIV asymptotics of orthogonal polynomials for the Gaussian weight with a jump at the edge, {Stud. Appl. Math.} {\bf 127} ({2011}), {67--105}.

\bibitem{XuZhao2019}
S.-X. Xu and Y.-Q. Zhao, Gap probability of the circular unitary ensemble with a Fisher-Hartwig singularity and the coupled Painlev\'{e} V system, arXiv: 1907.11509.
\end{thebibliography}
\end{document}